\documentclass[conference]{IEEEtran}
\IEEEoverridecommandlockouts
\makeatletter
\def\ps@headings{%
\def\@oddhead{\mbox{}\scriptsize\rightmark \hfil \thepage}%
\def\@evenhead{\scriptsize\thepage \hfil \leftmark\mbox{}}%
\def\@oddfoot{}%
\def\@evenfoot{}}
\makeatother
\usepackage{cite}

\ifCLASSINFOpdf
   \usepackage[pdftex]{graphicx}
\else
   \usepackage[dvips]{graphicx}
\fi
\usepackage[cmex10]{amsmath}
\usepackage[noend]{algorithmic}
\usepackage{array}
\usepackage{mdwmath}
\usepackage{mdwtab}
\usepackage{url}
\usepackage{boxedminipage}
\newtheorem{lemma}{Lemma}

\newtheorem{theorem}{Theorem}

\DeclareGraphicsExtensions{.eps}

\newcommand{\degree}{\mathrm{degree}}

\newcommand{\G}{\mathrm{G}}
\newcommand{\V}{\mathrm{V}}
\newcommand{\E}{\mathrm{E}}
\newcommand{\h}{\mathrm{h}}

\begin{document}
%\title{Bare Demo of IEEEtran.cls for Conferences}

%\title{Xheal+: The simplicity and power of edge-preserving self-healing}
 \title{Edge-preserving self-healing: keeping  network backbones densely connected} 

%\titlerunning{Xheal+: edge preserving self-healing}  % abbreviated title (for running head)

%\author{
%Atish Das Sarma\thanks{Google Research, Google Inc. - Mountain View, CA, USA - 94041.
% \hbox{E-mail}:~{\tt atish.dassarma@gmail.com}} 
% \and Amitabh Trehan\thanks{Faculty of Industrial Engineering and Management, Technion - Israel Institute of Technology, Haifa, Israel - 32000. 
%\hbox{E-mail}:~{\tt amitabh.trehaan@gmail.com} .\\ 
%Supported in part at the Technion by a fellowship of the Israel Council for Higher Education.}
%}

\author{\IEEEauthorblockN{Atish Das Sarma}
\IEEEauthorblockA{Google Research, Google Inc.\\
 Mountain View, CA, USA - 94041.\\
 E-mail:  atish.dassarma@gmail.com}
\and
\IEEEauthorblockN{Amitabh Trehan \thanks{Supported in part at the Technion by a fellowship of the Israel Council for Higher Education.}}
\IEEEauthorblockA{Faculty of Industrial Engineering and Management,\\
 Technion - Israel Institute of Technology,\\
  Haifa, Israel - 32000. \\
E-mail: amitabh.trehaan@gmail.com }
}


%Telephone: (800) 555--1212, Fax: (888) 555--1212}


\maketitle

\begin{abstract}
%\boldmath
%The abstract goes here.
Healing algorithms play a crucial part in distributed peer-to-peer networks where failures occur continuously and frequently. The general goal of self-healing distributed graphs is to maintain good connectivity throughout the network. This comes with the constraint that with every failure, one is allowed to only make bounded alterations locally. Several self-healing algorithms have been suggested in the recent literature [IPDPS'08, PODC'08, PODC'09, PODC'11] in a line of work that has yielded gradual improvements in the properties ensured on the graph. The competing requirements normally imposed are that of maintaining small degrees, while ensuring high connectivity in terms of shortest path dilation. In a recent work in PODC'11, an additional requirement on {\em expansion} was added, and an improved self-healing algorithm for maintaining the same was presented. This work motivates a strong general phenomenon of {\em edge-preserving} healing that aims at obtaining self-healing algorithms with the constraint that all original edges in the graph (not deleted by the adversary), be retained in {\em every} intermediate graph. This naturally further restricts the ability to add new edges during failures, due to the degree bound constraints. 

None of the previous algorithms, in their nascent form, are explicitly edge preserving. In this paper we show that the previous algorithms can be suitably modified with very simple changes such that all the previous properties are maintained, and in addition, the algorithms are {\em edge-preserving}. Towards this end, we present a general self-healing model that unifies the previous algorithms and shall hopefully be a definitive model. The main contribution of this paper is not in the technical complexity, rather in the simplicity with which the edge-preserving property can be ensured and the message that this is a crucial property with several benefits. In particular, we highlight the power of {\em edge-preserving} self-healing algorithms by showing that, almost as an immediate corollary, subgraph densities are preserved or increased. Maintaining density is a notion that is clearly motivated by the fact that in certain distributed networks, certain nodes may require and initially have a larger number of inter-connections (perhaps due to their larger bandwidth/communication requirements). It is vital that a healing algorithm, even amidst failures, respect these requirements; this is something that was not guaranteed by any of the previous algorithms. Our suggested modifications yield such subgraph density preservation as a by product. In addition, edge preservation helps maintain any subgraph induced property that is monotonic. Also, algorithms that are edge-preserving require minimal alteration of edges which can be an expensive cost in healing - something that has not been modeled in any of the past work. All the algorithms and proofs presented are simple and yet powerful enough to guarantee edge preservation in addition to all previous requirements.

\end{abstract}

\IEEEpeerreviewmaketitle

\section{Introduction}

With the advent of the internet, the growth of large communication networks, and the staggering rate of interaction and data exchange, the need for self maintaining distributed networks has grown tremendously. While such huge networks provide an excellent backbone for processing and exchanging data at a scale that was unimaginable before, they also call for radically improved and novel techniques that are efficient and fault tolerant at the same magnitude. At this scale, managing resources centrally is untenable. It is imperative that we design distributed and localized healing algorithms for failures, that achieve and maintain all desired global connectivity properties. The challenge lies in several dimensions: (a) localized distributed algorithms have the inability to look far into the network and so maintaining global properties is unclear, (b) several of the properties desired may themselves be conflicting on first sight such as ensuring upper bounds on degrees while maintaining low stretch and high expansion (c) the rate and nature of failures can be completely arbitrary (or adversarial) and thus, hard to predict, and finally yet very importantly (d) reacting to failures by deleting and adding new communication edges can be
%Amitabh: Our previous algorithms depend on edge changes not being 'extremely' expensive :)
% extremely 
expensive in both cost and time.  

This line of work adopts a responsive approach, in the sense that it responds to an attack (or component failure) by changing the topology of the network. This approach works irrespective of the
initial state of the network, and is thus orthogonal and complementary to traditional non-responsive techniques. In this setting, these papers seek to address the important and challenging problem of efficiently and responsively maintaining global invariants in a localized, distributed manner.

Healing algorithms play a crucial part in distributed peer-to-peer networks where failures occur continuously and frequently. Several self-healing algorithms have been suggested in the recent literaureture~\cite{SaiaTrehanIPDPS08, HayesPODC08, HayesPODC09,Amitabh-2010-PhdThesis,PanduranganPODC11} in a line of work that has yielded gradual improvements in the properties ensured on the graph. The competing requirements normally imposed are that of maintaining small degrees, while ensuring high connectivity in terms of shortest path dilation. In a recent work~\cite{PanduranganPODC11}, an additional requirement of {\em expansion} was added, and an improved self-healing algorithm for maintaining the same was presented. This work motivates a strong general phenomenon of {\em edge-preserving} healing that aims at obtaining self-healing algorithms with the constraint that all the original edges in the graph (not deleted by the adversary), be retained in {\em every} intermediate graph. This naturally further restricts the ability to add new edges during failures, due to the degree bound constraints. This is the primary focus of our paper and we motivate this further in the next paragraph.

Several algorithms have been designed to obtain self-healing in such distributed networks, and in a series of papers, some of these concerns have been addressed. One challenge in particular, however, has received little attention: namely, the cost of deleting and adding new edges. Past work has bounded this cost by restricting the number of edges that can be deleted or added. However, none of the works appeal to the concept of {\em minimizing} the edge {\em deletions} of the network, at any healing stage. The most recent algorithm, and arguably the best for this problem,  suffers from requiring deletion of edges with {\em every} possible local healing step. In this paper, we address this question by designing {\em edge-preserving} healing algorithms: we require that a self-healing algorithm not delete {\em any} edge that was originally present or was inserted during a node insertion  into the network. We however do allow deleting edges that were added in subsequent localized healing steps. While the edge-preserving property helps us minimize, or rather completely eliminate, the need to rewire communication paths originally present in any network, surprisingly, we can show much more. We are able to not only prove all the previous guarantees on all other requirements, but also show that edge-preservation immediately leads to several other very desirable global properties. Most notably, this helps us preserve the {\em density} of {\em every} current induced subgraph between nodes that were present originally in the network. In fact, edge-preserving self-healing ensures that every subgraph property that is edge-monotone i.e. non-decreasing with an increase in the number of edges does not suffer. Density is an important example of such a property.

Density is a very well studied notion in graphs and is an excellent measure of the inter-connectivity between groups of nodes. Density measures the {\em strength} of a set of nodes by the graph induced on them from the overall structure. The power of density lies in locally observing the strength of {\em any} set of nodes, large or small, independent of the entire network. Expansion, on the other hand measures the connection between a set of nodes and the rest of the network. Therefore, density constraints would be able to assess and better maintain the strength between specified (or all) subsets of nodes (even though there are exponential such sets) in a more robust manner. Consider for example a distributed network that contains a small set $S$ of $k$ nodes where these $k$ nodes are {\em hubs}, or central and crucial to the overall backbone of the network. It is conceivable and even likely that they would incur a larger communication interaction between them, and therefore demand larger connectivity structure, lower latency, and higher resilience to failures. Therefore, a peer to peer network would have a larger number of connections between these $k$ nodes as compared to any random set of $k$ nodes. It is then crucial that any self-healing algorithm retain this stronger interconnection within $S$, and not accidently compromise these connections for increasing connectivity at less desired places. Ensuring that the {\em density} between these nodes is preserved guarantees higher tolerance to failures, lower latencies and more efficient communication paths between the important set of nodes, namely $S$. 
%% Changed this; FG and FT are basically edge preserving.
Unfortunately not all previous healing algorithms were able to ensure this density requirement even if it were present in the original network; the healing stages would inadvertently compromise such hidden dense structures for more smoothed global connectivity properties (since expansion and other previous requirements did not capture this notion of induced density). The algorithms we present here are able to maintain (or even further improve) the density of {\em every} subset of nodes from the initial network, and through all stages of the self-healing steps.
 Surprisingly, we are able to achieve this without compromising on any of the other requirements. Therefore, in addition to all subgraph densities, our paper continues to satisfy degree constraints, preserve or even improve expansion, while maintaining low stretch, like the previous paper. 
 
 %Our algorithms are simple and surprisingly largely exploit the previous suggested algorithms, yet guarantee the strong {\em edge-preserving} property. The previous papers on this topic claim explicitly admit that they need the network to be {\em reconfigurable} in that edges may need to be continuously deleted and added - this limits their applicability. With the necessity to delete original edges eliminated, the applicability of our algorithms suggested in this paper extend to a larger class of networks including peer-to-peer, wireless mesh, and ad-hoc computer networks, and infrastructure networks (e.g. airline transportation). Most of these networks, while dynamic, would incur initiation and termination costs for new connections and dropping connections, whenever making an alteration in the networks. Our algorithms also help these networks retain their backbone connections and thereby minimize any reconstruction/rconfiguration costs.

Our algorithms are simple and  largely exploit the previous  algorithms, yet guarantee the strong {\em edge-preserving} property. The papers on this topic admit that they need the network to be {\em reconfigurable} in that edges may need to be continuously deleted and added.  However, continuous deletion and addition of edges can be costly.
% - this limits their applicability.
% The networks below are also reconfigurable.
With the necessity to delete original edges eliminated, the applicability of our algorithms suggested in this paper extends more effectively to a larger class of networks including peer-to-peer, wireless mesh, and ad-hoc computer networks, and infrastructure networks (e.g. airline transportation). Most of these networks, while dynamic, would incur initiation and termination costs for new connections and dropping connections, whenever making an alteration in the networks. Our algorithms also help these networks retain their backbone connections and thereby minimize any reconstruction/rconfiguration costs.

\medskip
%copied from FG
\noindent {\bf Our Model:}
%\Amitabh{Added the word bounded. Remove if doing for unbounded degrees.}\\
Our model - \emph{The General Self-healing model} is a generalization of the models used in ~\cite{PanduranganPODC11,Amitabh-2010-PhdThesis,HayesPODC09,HayesPODC08,SaiaTrehanIPDPS08}. We describe it here briefly.  We assume that the network is initially a connected graph over $n$ nodes.  An adversary repeatedly attacks the network. This adversary knows the network topology and our algorithm, and it has the ability to delete arbitrary nodes from the  network or insert a new node in the system which it can connect to any subset of  nodes currently in the system.  
%of a bounded number 
However, we assume the adversary is constrained in that in any time step it can only delete or insert a single node. In this context, the self-healing algorithm is supposed to fulfill certain success metrics, according to the particular requirements of the problem.
The detailed model is described in Section~\ref{sec: Xmodel}.

\medskip
\noindent {\bf Our Contributions.}
\begin{itemize}
\item We introduce edge preservation as a novel consideration for healing algorithms, motivated by the cost associated with switching physical network communication lines. We show edge preservation not only reduces the cost but also results in desirable network structural consequences such as preserving or improving various monotonic graph properties. In particular, we motivate and use density as a running example for an interesting monotonic property. We then consider density as a novel constraint for the design of self-healing algorithms.
\item We generalize and strengthen previous algorithms to consider edge-preservation, and in particular maintaining or increasing all subgraph densities, as a constraint to the algorithmic problem. This enforces the constraint of maintaining any strong connectivity structure between hub nodes, or nodes with high bandwidth/communication.
% requires, a requirement that was neither consider nor satisfied by the previous healing algorithms.
\item Finally, we present proofs for the unified algorithm presented in this paper by showing strong guarantees on the density constraint, as well as show how the previously considered measures, such as connectivity, diameter, degree constraint, network stretch, and expansion, continue to hold without being compromised. Our proofs are simple and the algorithm and proof essentially fall out from a fairly simple generalization of the previous techniques by a small modification that enforces the edge-preserving property.
\end{itemize}
 
\noindent {\bf Related Work:} 
%\Amitabh{Fix the first paras}
This work builds upon previous works of self-healing. It investigates the algorithms discussed in~\cite{PanduranganPODC11,Amitabh-2010-PhdThesis,HayesPODC09,HayesPODC08,SaiaTrehanIPDPS08}; In particular, we introduce an edge-preserving (formally defined later) version of the algorithm \emph{Xheal}  and show that this new version self-heals subgraph density.

These works show a progressive increase in the number of properties self-healed and increasing sophistication in the techniques used. The earliest works~\cite{SaiaTrehanIPDPS08,BomanSAS06} maintained connectivity while ensuring low degree increase. Later, Hayes, Rustagi, Saia and Trehan introduced the \emph{Forgiving Tree}~\cite{HayesPODC08}, which is a tree maintenance algorithm, that while being much more efficient also maintained the diameter of the network. However, it does not handle insertion (this is still an open problem in the tree maintenance setting). Hayes, Saia and Trehan later introduced the \emph{Forgiving Graph}~\cite{HayesPODC09}, which improved upon Forgiving Tree by not only handling insertions but also maintaining network stretch (which is a stronger property than network diameter). All the above algorithms use tree like structures for self-healing; Pandurangan and Trehan instead use expander structures in  Xheal~\cite{PanduranganPODC11} and show that they can also self-heal Network expansion in addition to the previous properties. Our work augments these previous works by adding a new desirable property for self-healing algorithms,  adding subgraph density to the list of properties maintained, analyzing the previous algorithms for these properties and modifying them (when required).

Further, the problem of finding densest subgraphs is well-studied. Finding a maximum density subgraph on an undirected graph can be solved in polynomial time~\cite{G84, L}. However, the problem becomes NP-hard when a size restriction is enforced. In particular, finding a maximum density subgraph of size exactly $k$ is NP-hard~\cite{AHI, FKP} and no approximation scheme exists under a reasonable complexity assumption~\cite{K}. Khuller and Saha~\cite{KS} considered the problem of finding densest subgraphs with size restrictions and showed that these are NP-hard. Khuller and Saha ~\cite{KS} and also Andersen and Chellapilla ~\cite{AC} gave constant factor approximation algorithms. There are several more papers related to finding dense subgraphs, both theoretical and practical, but we avoid listing a comprehensive survey of this literature and refer the reader to the references in the aforementioned papers. The nice aspect of our self-healing algorithms approach is that it completely bypasses the problem of actually computing dense subgraphs; the density of any subgraph is simply retained (or increased) as an immediate consequence of the edge-preserving property.

\subsection{Preliminaries}
Let $\G =(\V,\E)$ be an undirected graph and $S \subseteq \V$ be a set of nodes.

\noindent \textbf{Graph Density:}
The density of a graph $G(V, E)$ is defined as $|E|/|V|$.

\noindent \textbf{SubGraph Density:}
The density of a subgraph defined by a subset of nodes $S$ of $V(G)$ is defined as its induced density. We will use $den(S)$ to denote the density of the subgraph induced by $S$. Therefore, $den(S) = \frac{|E(S)|}{|S|}$. Here $E(S)$ is the subset of edges $(u, v)$ of $E$ where $u\in S$ and $v\in S$.  In particular, when talking about the density of a subgraph defined by set of vertices $S$ induced on $G$, we use the notation $den_G(S)$. However, when clear from context, we omit the subscript $G$.

\noindent \textbf{Edge Expansion:} 
We denote $ \overline{S} = V - S$. Let $|\E|_{S, \overline{S}} = \{ (u,v) \in \E | u \in S, v \in \overline{S} \}$ be the number of edges crossing the cut $(S, \overline{S})$.  We define the {\em volume} of $S$ to be the sum of the degrees
of the vertices in $S$ as $vol(S) = \sum_{x \in S}degree(x)$. The edge expansion of the graph $\h_G$ is defined as, 
 $
  \h_G =  {\mathrm{min}_{ |S| \le |\V|/2}}  \frac{|\E|_{S,\overline{S}}}{|S|}   
 $, 

\section{General Self-Healing Model}
\label{sec: Xmodel}

%\floatname{algorithm}{Model}

%\begin{algorithm}[h!]
\begin{figure}[h!]
\caption{The Basic Self-healing(Node Insert, Delete and Network Repair) Model -- Distributed View.}
\label{algo:model-2}
\begin{boxedminipage}{0.5\textwidth}
\begin{algorithmic}
\STATE Each node of $G_0$ is a processor.  
\STATE Each processor starts with a list of its neighbors in $G_0$.
\STATE Pre-processing: Processors may send messages to and from
their neighbors.
\FOR {$t := 1$ to $T$}
\STATE Adversary deletes or inserts a node $v_t$ from/into $G_{t-1}$, forming $U_t$.
\IF{node $v_t$ is inserted} 
\STATE The new neighbors of $v_t$ may update their information and send messages to and from
their neighbors.
\ENDIF
\IF{node $v_t$ is deleted} 
\STATE All neighbors of $v_t$ are informed of the deletion.
\STATE {\bf Recovery phase:}
\STATE Nodes of $U_t$ may communicate (synchronous/asynchronous, in parallel) 
with their immediate neighbors.  These messages are never lost or
corrupted, and may contain the names of other vertices.
\STATE During this phase, each node may insert edges
joining it to any other nodes as desired. 
Nodes may also drop edges from previous rounds if no longer required.
\ENDIF
\STATE At the end of this phase, we call the graph $G_t$.
\ENDFOR
\vspace{10pt}
\hrule
\STATE
\STATE {\bf Success metrics:} 
%Minimize the following ``complexity'' measures:\\
%Consider the graph  $G'_t$ which is the graph, at timestep $t$, consisting solely of the original nodes (from $G_0$) and insertions without regard to deletions and healings. 
%Graph $G'_{t}$ is $G'$ at timestep $t$ (i.e. after the $t^{\mathrm{th}}$ insertion or deletion).
%Graph $G'_{t}$ is $G'$ at timestep $t$ which is equivalent to $G'_{t'}$ where the $t' \le t$ is
%the timestep at which the latest insertion on or before $t$ occured. 
 \begin{itemize}
\item{\bf Maintaining properties:} Maintain certain well stated invariants/ minimize certain(local/global) ``complexity'' measures.
\item{\bf Recovery time:} The maximum total time for a recovery round,
assuming it takes a message no more than $1$ time unit to traverse any edge and we have unlimited local computational power at each node.
\item{\bf Communication complexity:} Number of messages used for recovery.
\end{itemize}
\end{algorithmic}
\end{boxedminipage}
\end{figure}
%\end{algorithm}

\begin{figure}[h!]
\caption{Example Properties/ ``Complexity'' measures for the basic self-healing model}
\label{algo: sh-complexitymeasures}
\begin{boxedminipage}{0.5\textwidth}
\begin{algorithmic}
\STATE
Consider the graph  $G'_t$ which is the graph, at timestep $t$, consisting solely of the original nodes (from $G_0$) and insertions without regard to deletions and healings. 

 \begin{enumerate}
 
 \item{\bf Connectivity.} If $G'_t$ is connected, so is $G_t$.
 \item{\bf Degree increase.}  $\max_{v \in G_t} \frac{\degree(v,G_t)}{ \degree(v,G'_t)}$.
 \item{\bf Density.} $den_{G_t}(S) \geq den_{G'_t}(S)$.
\item {\bf Edge expansion.} $h(G_t) \ge min(\alpha,\beta h(G'_t))$; for constants $ \alpha , \beta > 0$.
\item{\bf Diameter.} diameter($G_t$)/ diameter($G'_t$).
\item {\bf Network stretch.} $\max_{x, y \in G_{t}} \frac{dist(x,y,G_{t})}{dist(x,y,G'_{t})}$, where, for a graph $G$ and nodes $x$ and $y$ in $G$, $dist(x,y,G)$ is the
length of the shortest path between $x$ and $y$ in $G$.
%\item{\bf Communication per node.} The maximum number of bits sent by a single node in a single recovery round.
% \tom{Want to modify this or omit?}
\end{enumerate}

\end{algorithmic}
\end{boxedminipage}
\end{figure}

%\Amitabh{This may need modification/ editing}

This model was introduced in~\cite{HayesPODC09,Amitabh-2010-PhdThesis}. Somewhat similar models were also used in~\cite{PanduranganPODC11,HayesPODC08,SaiaTrehanIPDPS08}.
We now describe the details. 
Let $\chi$ be a self-healing algorithm.
 Let $G = G_0$ be an arbitrary graph on $n$ nodes, which represent processors in a distributed network.  In each step, the adversary either deletes or adds a node.  After each deletion, the algorithm gets to add some new edges to the graph, as well as deleting old ones.  At each insertion, the processors follow a protocol to update their information.
 $\chi$'s goal is to maintain a certain set of properties e.g. those given in Figure~\ref{algo: sh-complexitymeasures}. At the same time, the algorithm wants to minimize the resources spent on this task,  which usually also includes keeping node degree small.  

%The algorithm's goal is to maintain connectivity in the network, while maintaining good expansion properties and  keeping the distance between the nodes small.  At the same time, the algorithm wants to minimize the resources spent on this task, including keeping node degree small.  

%We seek an algorithm
%which gives performance guarantees under these metrics for each of the  possible insertion and deletion orders.

Initially, each processor only knows its neighbors in $G_0$, and is unaware of the structure of the rest of $G_0$.
After each deletion or insertion, only the neighbors of the deleted or inserted vertex are informed that
the deletion or insertion has occurred. After this, processors are allowed to communicate (synchronously or asynchronously depending on the constraints) by sending a limited number
of messages to their direct  neighbors.  We assume that these messages are always sent and received successfully.  The
processors may also request new edges be added to the graph.
We make sure that no
other  vertex is deleted or inserted until the end of this round of computation and communication has concluded.
%To make this assumption more reasonable, the per-node communication cost should be very small in $n$ (e.g. at most logarithmic).
%$O(1)$ bits, and should 
%moreover be parallelizable so that the entire protocol can be completed in $O(1)$ time. 

We also allow a certain amount of pre-processing to be done before the first attack occurs.   For example, we assume that all nodes know the address
of all the neighbors of its neighbors (NoN).  Our full model is described in Figure~\ref{algo:model-2}.

\section{Edge-preserving Self Healing}
\label{edge-preserving}
 Here, we introduce edge-preservation, which we contend is a strongly desirable property for self-healing. A self-healing algorithm is edge-preserving in our model if the original edges and those inserted by the adversary are never deleted by the algorithm. More formally, we state:\\ 
%\begin{description}
\noindent \textbf{Edge Preserving:}  A self-healing algorithm $\chi$ is edge-preserving in the general self-healing model (Figure~\ref{algo:model-2}), if we have that,  for all $u, v\in V(G_t)$, if $(u, v) \in E(G'_t)$, then $(u, v) \in E(G_t)$. 
%\end{description}

We also define the notion of an edge-monotonic property as follows:\\
\noindent \textbf{Edge-monotonic graph property/function:} Given a graph $G(V, E)$, and a subgraph $S\subseteq V(G)$, a subgraph property/function $f_G:S\rightarrow [0,\infty)$ is said to be edge-monotonic or edge-monotonically non-decreasing if for any two graphs $G_1(V, E_1)$, $G_2(V, E_2)$ with $E_1(S)\subseteq E_2(S)$, we have $f_{G_1}(S)\leq f_{G_2}(S)$. Further, the property is said to be edge-monotonically increasing if for $E_1(S)\subset E_2(S)$, we have $f_{G_1}(S) < f_{G_2}(S)$.

It is quite straightforward to maintain edge-monotonic graph properties once edge-preservation is achieved. We show this more formally in section~\ref{sec: xanalysis} for the algorithm xheal+ (described later) and the edge-monotonic property of density.

%\section{The algorithm
\section{Xheal+}
\label{sec: xh-algorithm}

Here, we describe \emph{xheal+}, which is the algorithm xheal~\cite{PanduranganPODC11} modified to make it edge-preserving.  The main difference from xheal is that we allow multiple `colorings' (explained more formally later) for a single edge in xheal+. This enables us to detect if the edge was originally present or inserted by the adversary in the graph and has been recolored by the algorithm. At certain points in its execution, xheal removed edges from the graph; in xheal+, if the edge was not an `original' edge, we delete the edge as in xheal, but otherwise, we simply remove the required label/color from the edge without deleting it.  The algorithm is summarized in Figure~\ref{algo: repairbyexpander}; to make it easy to understand the modifications from xheal, we have added the symbol \emph{+} to the lines we have changed. We have also rewritten some of the subroutines more clearly and added the subroutines given in Figures~\ref{algo: mergeclouds}, \ref{algo: maketopology}, \ref{algo: markedges}, and~\ref{algo: deleteedges}.

 Here, we will not describe the details of the algorithm which are already given in~\cite{PanduranganPODC11}. However, we shall very briefly summarise the algorithm for completeness and explain in more detail the enhancements in xheal+. 
 %The algorithm is summarized as Algorithm~\ref{algo: repairbyexpander}.
  Let $\kappa$ be a fixed parameter that is implementation dependent. For the purposes of this algorithm, we assume the existence of a $\kappa$-regular expander with edge expansion $\alpha > 2$. 
 To describe the algorithm, in xheal+, we associate a set of colors  (as opposed to a single color in xheal) with each edge of the graph i.e. for an edge $e$, $e.color$ is a set of colors associated with an edge. Each color associates a property or functionality with an edge.  We  assume that the original edges of $\G$ and those added by the adversary are all colored {\bf black} initially i.e. for such an edge $e$, $e.color = \{black\}$.  If $(u,v)$ is a black (colored) edge, we say that $v$($u$) is a black (colored) neighbor of $u$($v$). In xheal+, the algorithm can later add functionality to the edge (add as part of primary or secondary clouds, as described later) by adding new colors to the set, and remove the functionality by simply removing that color.
 
  At any time step, the adversary can add a node (with its incident edges) or delete a node (with its incident edges). Addition is straightforward, the algorithm takes no action and the added edges simply get colored {\bf black} (Notice the edge did not exist before, so this will be its first color; in a model where edges may be adversarially added to previously existing nodes, this may not hold). The self-healing algorithm is mainly concerned with what edges to add when a node is deleted and this is done based on the colors of the edges deleted as well as on other factors. In brief, the neighbors of the deleted node may be all black i.e. $e.color = black$ for all edges $e$ deleted in this step, or not all black. If they are all black, the neighbors reconnect as a 'primary' cloud(Figure~ \ref{algo: makecloud}) i.e. as a $\kappa$-regular expander or as a clique if number of neighbors are less than $\kappa$ (Figure~\ref{algo: maketopology}). This cloud has its own color (e.g. it could be  the label or ID of the deleted node); if a required edge does not exist, a new one is created and it takes this color, but if the edge already exists, in xheal+, this color is added to the color set of the edge (Figure~\ref{algo: maketopology}). If the neighbors of the deleted node are not all black, this implies that the deleted node was part of at least one primary cloud (a node participates exactly once in a primary cloud for each of its deleted neighbor); we fix these primary clouds by redrawing (deleting/adding/reusing) some edges (Figure~ \ref{algo: fixprimary}) to restore them to be expanders. In xheal+, since you are not allowed to remove original edges, there may be a need or advantage in doing this more efficiently reusing existing edges. Now, we select a `free' node (explained later) from each of these primary clouds and construct a secondary cloud (Figure~\ref{algo: makesecondary}). A free node is simply a node which is not taking part in any secondary cloud, thus, a node can take part in at most one secondary cloud.   When a node is a part of a secondary cloud, it is called a bridge node for the primary cloud it represents. Each Secondary cloud also has its own distinct color and this gets added onto the edge if it's also used for secondary cloud duties. On deletion of a node, secondary cloud edges may also be lost, in which case we repair this cloud too (Figure~\ref{algo: fixsecondary}). We do this by finding a new free node for the primary cloud that lost the bridge node (i.e. the deleted node), if need be, by borrowing from neighboring primary clouds. However, some times this may not be possible, in which case, we merge all the primary and secondary clouds effected and make a new primary cloud from all their nodes (Figure~\ref{algo: mergeclouds}). Merging is an expensive operation but since it will not happen often, its cost is amortized over previous operations. 
  
   Edges may be deleted in the cloud fixing and merging operations, thus, these operations need to be edge preserving.  Also, we have to ensure the nodes can communicate while the reconstruction is underway.
   To ensure this, we have added the algorithms  \textsc{MarkEdges()}(Fig.~\ref{algo: markedges}) , \textsc{MakeTopology()}(Fig.~\ref{algo: maketopology}), and \textsc{DeleteEdges}(Fig.~\ref{algo: deleteedges}). The algorithm  \textsc{MarkEdges()} prepares the clouds for construction or repair by removing the cloud's color from the edges and marking those edges which can be safely removed (If an edge was originally present or added by the adversary, it will have the color black in its set and thus will not be marked). Notice that the edges themselves cannot be removed at this stage since these edges will form the network of communication for the present round of repair. 
   \textsc{MakeTopology} is used to make the cloud edges; it checks for existence of a required edge for the cloud being constructed, if that edge already exists, it simply adds the color of the cloud to the color set of the edge.
   % Notice in case of the cloud fixing operations, that this is now a simpler operation since the Marking algorithm has already removed the cloud color
    If an edge does not exist between the two nodes, a new edge is constructed and the edge's color set is initialized by the color of the cloud. 
    \textsc{DeleteEdges} is the subroutine called after a new cloud is in place to clean up. It checks all the marked edges to see which ones have no color left. This means these edges were not reused and they are also not original edges and can thus be deleted.
   
%% During the repair, a node's degree can be higher since it may retain edges from the cloud under repair and new edges too before the deleteedges operation.%%

\begin{figure}[h!]
\caption{\textsc{Xheal}($\G, \kappa$)}
%\begin{boxedminipage}{0.5\textwidth}
\begin{algorithmic}[1]

 \IF{node $v$ inserted with  incident edges}
%\STATE The inserted edges are colored black.
\STATE \textbf{+ $\forall$ inserted edges $e$, $e.color \leftarrow \{black \}$.}
\ENDIF

\IF{node $v$ is deleted}
\label{algoline: xh-black} \IF{\textbf{+} $\forall$ deleted edges $e$, $e.color  = \{black\} $} 
%\COMMENT{Case 1}
\STATE   \textsc{MakeCloud}($BlackNbrs(v), primary, Clr_{new}$)
\ELSIF{deleted colored edges are all primary}
%[Case 2.1]
\STATE Let $C_1, \dots, C_j$ be primary clouds that lost an edge
\STATE  \textsc{FixPrimary}($[C_{1}, \dots, C_{j}]$)
\STATE \textsc{MakeSecondary}($[C_1, \dots, C_j] \cup BlackNbrs(v)$)
\ELSE
% \Comment{Case 2.2}
\STATE Let $[C_1, \dots, C_j] \leftarrow$ primary clouds of $v$; $F \leftarrow$ secondary cloud of $v$; $[U] \leftarrow$  $Clouds(F) \setminus [C_1, \dots, C_{j}]$, $[C_1, \dots, C_{j'}] \leftarrow F \cap [C_1, \dots, C_{j}]  $  
\STATE  \textsc{FixPrimary}($[C_{1}, \dots, C_{j}]$)
\STATE   \textsc{FixSecondary}($F, v$)
% Fixed the line below: should C_j' instead of C_{j' + 1}
\STATE \textsc{MakeSecondary}($[C_{j'}, \dots, C_j] \cup BlackNbrs(v)$)
\ENDIF
 \ENDIF
\end{algorithmic}
\label{algo: repairbyexpander}
\end{figure}

\begin{figure}[h!]
\caption{\textsc{FixPrimary}($[C]$)}
%\begin{boxedminipage}{0.5\textwidth}
\begin{algorithmic}[1]
\STATE \textbf{+} \textsc{MarkEdges}($[C]$)
\FOR{each cloud $C_i \in [C]$ }
\STATE  \textsc{MakeCloud}($C_i, primary, Color(C_i)$)   
\ENDFOR
\STATE \textbf{+} \textsc{DeleteEdges}($[C]$)
\end{algorithmic}
%\end{boxedminipage}
%\end{figure}
\label{algo: fixprimary}
\end{figure}

\begin{figure}[h!]
\caption{\textsc{MakeSecondary}($[C]$)}
%\begin{boxedminipage}{0.5\textwidth}
\begin{algorithmic}[1]
\FOR{each cloud $C_i \in [C]$ }
\IF{$FrNode_i$ = \textsc{PickFreeNode}($C_i$) == NULL} 
%\STATE \textbf{+} \textsc{RemoveClouds([C])}
\STATE \textbf{+}  \textsc{MergeClouds}($[C]$)   
% \COMMENT{Free node not available, reconstruct}
\STATE Return
\ENDIF
\ENDFOR
\STATE  \textsc{MakeCloud}($\bigcup FrNode_i\ \forall C_i \in [C]$, secondary, $Clr_{new}$)
\end{algorithmic}
%\end{boxedminipage}
%\end{figure}
\label{algo: makesecondary}
\end{figure}

\begin{figure}[h!]
\caption{\textsc{FixSecondary}($F, v$)}
%\begin{boxedminipage}{0.5\textwidth}
\begin{algorithmic}[1]
\IF{$v$ is a bridge node of $C_i$ in $F$} 
\IF{$FrNode_i$ = \textsc{PickFreeNode}($C_i$) == NULL} 
\STATE \textbf{+}  \textsc{MergeClouds}($F$)   
% \COMMENT{Free node not available, reconstruct}
\ELSE
\STATE \textbf{+} \textsc{MarkEdges}($F$)
\STATE  \textsc{MakeCloud}($FrNode_i \cup BridgeNode(C_j)\ \forall C_j \in [C]$, secondary, Color(F))    
\STATE \textbf{+} \textsc{DeleteEdges}($F$)
\ENDIF
\ENDIF
\end{algorithmic}
%\end{boxedminipage}
%\end{figure}
\label{algo: fixsecondary}
\end{figure}

\begin{figure}[h!]
\caption{\textsc{MakeCloud}($[V], Type, Clr$)}
\begin{algorithmic}[1]
\IF {$|V| \leq \kappa+1$}
%\STATE Make clique among $[V]$
\STATE \textbf{+} \textsc{MakeTopology($[V], Type, Clr, clique$)}
\ELSE
%\STATE Make $\kappa$-reg expander among $[V]$ of edge $(Type, Clr)$
\STATE \textbf{+} \textsc{MakeTopology($[V], Type, Clr,\kappa-reg expander $)}
\ENDIF
\end{algorithmic}
\label{algo: makecloud}
\end{figure}

\begin{figure}[h!]
\caption{\textbf{+} \textsc{MergeClouds}($[C]$)}
\begin{algorithmic}[1]
\STATE  \textsc{MarkEdges}($[C]$)
\STATE \textsc{MakeCloud}($Nodes([C], Primary, Clr_{new}$)
\STATE \textsc{DeleteEdges}($[C]$)
\end{algorithmic}
\label{algo: mergeclouds}
\end{figure}

\begin{figure}[h!]
\caption{\textsc{\textbf{+} MakeTopology}($[V], Type, Clr, Top$)}
\begin{algorithmic}[1]
\STATE Design graph $T([V], E)$ of topology $Top$ \COMMENT{The nodes make the 'blueprint', then implement it}
\FOR{ each edge $e \in E$}
\IF{$e$ existed previously}
\STATE $e.color \leftarrow e.color \cup Clr$, $e.type \leftarrow Type$ \COMMENT{Reuse Edge}
\ELSE
\STATE Make new edge $e$; set $e.color \leftarrow Clr$, $e.type \leftarrow Type$
\ENDIF
\ENDFOR
\end{algorithmic}
\label{algo: maketopology}
\end{figure}

\begin{figure}[h!]
\caption{\textbf{+} \textsc{MarkEdges}($[C]$)}
\begin{algorithmic}[1]
\FOR{each cloud $C_i \in [C]$ }
\FOR{ each edge $e \in edges(C_i)$}
\STATE $e.color \leftarrow e.color \setminus color(C_i)$ \COMMENT{Remove edge from Cloud.}
\IF{$e.color = NULL$}
\STATE Mark $e$ for deletion \COMMENT{Not Original Edge; possibly remove at end of phase}
\ENDIF
\ENDFOR
\ENDFOR
\end{algorithmic}
\label{algo: markedges}
\end{figure}

\begin{figure}[h!]
\caption{\textbf{+} \textsc{DeleteEdges}($[C]$)}
\begin{algorithmic}[1]\FOR{all marked edges $e \in edges([C])$}
 \IF{$e.color = NULL$} %\COMMENT{Need to check in case edge has been reused\\}
  \STATE Delete edge \COMMENT{safe to remove edge now.}
  \ELSE
  \STATE Unmark edge
\ENDIF
\ENDFOR
\end{algorithmic}
\label{algo: deleteedges}
\end{figure}

\begin{figure}[h!]
\caption{\textsc{PickFreeNode}()}
%\begin{boxedminipage}{0.5\textwidth}
\begin{algorithmic}[1]

\STATE Let a Free node be a primary node without secondary duties
\IF{Free node in my cloud}
\STATE Return Free node
\ELSE
\STATE Ask neighbor clouds; if a free node found, return node, else return NULL
\ENDIF
\end{algorithmic}
%\end{boxedminipage}
%\end{figure}
\label{algo: pickfreenode}
\end{figure}

\section{Analysis of Xheal+}
\label{sec: xanalysis}

Our main claim is that xheal+,  the edge-preserving version of xheal, improves xheal
by giving not only  the same self-healing guarantees but also self-healing monotonic subgraph induced properties such as subgraph density.  Recall the definition of an edge-monotonic graph property/function (Section~~\ref{edge-preserving}).
%We briefly digress to recall what  what we mean by edge-monotonic properties (as defined in Section~\ref{edge-preserving}) and then return to our running example of density. 
%\noindent \textbf{Edge-monotonic graph property/function:} Given a graph $G(V, E)$, and a subgraph $S\subseteq V(G)$, a subgraph property/function $f_G:S\rightarrow [0,\infty)$ is said to be edge-monotonic or edge-monotonically non-decreasing if for any two graphs $G_1(V, E_1)$, $G_2(V, E_2)$ with $E_1(S)\subseteq E_2(S)$, we have $f_{G_1}(S)\leq f_{G_2}(S)$. Further, the property is said to be edge-monotonically increasing if for $E_1(S)\subset E_2(S)$, we have $f_{G_1}(S) < f_{G_2}(S)$.
It is easy to see that $f_G(S)$ as the density function $den_G(S) = \frac{|E(S)|}{|S|}$ is an edge-monotonically increasing graph property. We now return to stating the main theorem of our paper in full generality and then turn to proving it formally. The proof only focuses on density, for ease of presentation. It is self-evident why similar guarantees translate to various monotonic graph properties. 

\begin{theorem}
\label{th: xmain}
The algorithm xheal+:
 \begin{enumerate}
  \item provides the same self-healing guarantees as xheal.
   \item further, for graph $G_t$(present graph) and graph $G'_t$ (of only original and inserted edges), at any time t, where a timestep is an insertion or deletion followed by healing:
For all $S \subseteq V(G_t)$, and any edge-monotonically non-decreasing function $f_G$, we have $f_{G_t}(S)\geq f_{G'_t}(S)$. In particular, we have $den_{G_t}(S) \geq den_{G'_t}(S)$.
   \end{enumerate}
\end{theorem}

Expanding the above theorem using~\cite{PanduranganPODC11}, we get the main theorem on the guarantees that Xheal+ provides on the topological properties of the healed graph, assuming  Xheal+ is able to construct a $\kappa$-regular expander (deterministically), for a fixed constant $\kappa > 0$.

\begin{theorem}
\label{th:main}
 For graph $G_t$(present graph) and graph $G'_t$ (of only original and inserted edges), at any time t, where a timestep is an insertion or deletion followed by healing:
 \begin{enumerate}
  \item \label{thpart: stretch} For any two nodes $u,v \in G_t$,  $\delta_{G_t}(u,v) \le \delta_{G'_t}(u,v) O(\log n)$, where $\delta(u,v)$ is the shortest path between $u$ and $v$,
  and $n$ is the number of nodes in $G_t$.
   \item  \label{thpart: exp} $h(G_t) \ge min(\alpha, h(G'_t))$, for some fixed constant  $\alpha \ge 1$, where $h(G)$ is the edge expansion of a graph $G$
   \item  \label{thpart: cond} $\lambda(G_t) \ge min(A, B)$, 
   where $\lambda(G_t)$ is the second smallest eigenvalue of the Laplacian of $G_t$,  $A = \Omega\left(\lambda(G'_t)^2 \frac{d_{min}(G'_t)}{(\kappa)^2(d_{max}(G'_t))^2}\right)$, $B = \Omega\left(\frac{1}{(\kappa d_{max}(G'_t))^2}\right)$, 
    $d_{min}(G'_t)$ and $d_{max}(G'_t)$ are the minimum
   and maximum degrees of $G'_t$.
    \item  \label{thpart: deg} For all $x \in G_t$, $degree_{G_t}(x) \le \kappa.degree_{G'_t}(y) + \kappa$, for a fixed constant $\kappa> 0$.
     \item  \label{thpart: density} For all $S \subseteq V(G_t)$, and any edge-monotonically non-decreasing function $f_G$, we have $f_{G_t}(S)\geq f_{G'_t}(S)$. In particular, we have $den_{G_t}(S) \geq den_{G'_t}(S)$.
    \end{enumerate}
\end{theorem}

\begin{IEEEproof}
Observe that a graph healed by xheal+ can only have more edges than one healed by xheal, since certain edges are prohibited to be deleted in xheal+. Also, they are identical algorithms in other aspects. Since both stretch can only be lower and expansion can only be higher if the edges are more,  Parts \ref{thpart: stretch} and \ref{thpart: exp} follow (since these are proven to hold for xheal). Part~\ref{thpart: cond} is not affected since it is a statement bounding the conductance of the graph in relation to the minimim degree and maximum degree of nodes.  We only need to worry about parts~\ref{thpart: deg} and \ref{thpart: density} . These follow from Lemmas \ref{lemma: degree} and \ref{lemma: density-lb}.

\end{IEEEproof}

\subsection{Sub-graph Density Analysis via Edge Preserving Property}

In Theorem~\ref{th:main}, we claim that for all $S \subseteq V(G_t)$, $den_{G_t}(S) \leq den_{G'_t}(S)$. We initiate the proof of this lemma via a significantly stronger, and independently desirable, property that we call edge-preservering property: 

%\begin{description}
\noindent \textbf{Edge Preserving:} For all $u, v\in V(G_t)$, if $(u, v) \in E(G'_t)$, then $(u, v) \in E(G_t)$. 
%\end{description}

We state and prove this in the following lemma.

\begin{lemma}
\label{lemma: edgepreserving}
xheal+ is edge-preserving.
\end{lemma}
\begin{IEEEproof}
The proof follows directly from the algorithm. The algorithm explicitly makes sure that it never deletes an edge that was present in the original graph, or was inserted by an adversary. This is clear from 
the algorithms  \textsc{MarkEdges()}(Fig.~\ref{algo: markedges}) , \textsc{MakeTopology()}(Fig.~\ref{algo: maketopology}), and \textsc{DeleteEdges}(Fig.~\ref{algo: deleteedges}), which are the only subroutines ultimately responsible for adding or deleting edges. 
The algorithm  \textsc{MarkEdges()} marks the present edges which may possibly be affected by the reconstruction and those safe for deletion if not reused. 
   \textsc{MakeTopology} reuses existing edges if they are part of the new clouds being formed or constructs new ones initializing the edge's color  vector.  % Notice in case of the cloud fixing operations, that this is now a simpler operation since the Marking algorithm has already removed the cloud color
  %  If an edge does not exist between the two nodes, a new edge is constructed and the edge's color set is initialized by the color of the cloud. 
    \textsc{DeleteEdges} is the subroutine called after a new cloud is in place to clean up. It checks all the marked edges to determine the edges which are safe to be deleted by simply checking if their color set is empty. If the edge has no color i.e. empty color set, this means it was not an original or adversary inserted edge and has not been reused, and thus, can be safely deleted.
%The key to notice is that the algorithm, even when rewiring edges or modifing subgraphs, never deletes an edge that was present in the original graph, or was inserted by an adversary.
\end{IEEEproof}

Now, we prove our main lemma about subgraph density. We are mainly concerned with lower bounding the density but in the following section,  we also put in the upper bound for completeness.

\begin{lemma}
\label{lemma: density-lb}
%For all $S \subseteq V(G_t)$, $den_{G'_t}(S) \leq den_{G_t}(S) \leq max(\kappa.den_{G'_t}(S), \lfloor \frac{\kappa + 1}{2} \rfloor)$.
For all $S \subseteq V(G_t)$, $den_{G'_t}(S) \leq den_{G_t}(S) \leq \frac{\kappa.\sum_{i=1}^{|S|} deg_{G'_t}(x_i)}{|S|} + \frac{\kappa}{2}$.
\end{lemma}
\begin{IEEEproof}
Subgraph density is defined as  $den(S) = \frac{|E(S)|}{|S|}$ for an induced subgraph of a subset  $S$ of $V(G)$. Consider any subset $S \subseteq V(G_t)$. Since xheal+ is edge-preserving and $G'_t$ contains only original or adversary inserted edges, $E_{G'_t}(S) \subseteq  E_{G_t}(S)$ and therefore $den_{G'_t}(S) \leq den_{G_t}(S) $. %The proof follows directly from Lemma~\ref{lm: edgepreserving} and the definition of subgraph density...
\end{IEEEproof}

\subsection{Upper bounds on Density}
 In our paper we have been mainly concerned with making sure the density does not decrease. Here, for completeness, we also study how much the density can increase. 
\begin{lemma}
\label{lemma: density-ub}
%For all $S \subseteq V(G_t)$, $den_{G'_t}(S) \leq den_{G_t}(S) \leq max(\kappa.den_{G'_t}(S), \lfloor \frac{\kappa + 1}{2} \rfloor)$.
For all $S \subseteq V(G_t)$, $den_{G_t}(S) \leq den_{G'_t}(S) +  \frac{\kappa.\sum_{i=1}^{|S|} deg_{G'_t}(x_i), x_i \in S}{2|S|} + \frac{\kappa}{2}$.
\end{lemma}
\begin{IEEEproof}
From Lemma ~\ref{th:main}, part \ref{thpart: deg} (or Lemma~\ref{lemma: degree}), it follows that for the subset $S$, if each node had the maximum degree increase and all the added edges were part of the induced subgraph $G_t(S)$, then $|E_{G_t}(S)| = |E_{G'_t}(S)| +\sum_{i=1}^{|S|} (\frac{\kappa}{2} deg(x_i) + \frac{\kappa}{2})$. In the equation,  $\frac{\kappa}{2}$ comes from the fact that each edge contributes to the degree increase of two nodes.  Dividing by $S$, the lemma follows.  Notice that the worst case comes when there was no edge between any node of $S$ in $G'_t$ i.e. $ |E_{G'_t}(S)| = 0$. 
\end{IEEEproof}

Let us consider the change in density of the present graph. 
\begin{lemma}
\label{lemma: density-graph-ub}
%For all $S \subseteq V(G_t)$, $den_{G'_t}(S) \leq den_{G_t}(S) \leq max(\kappa.den_{G'_t}(S), \lfloor \frac{\kappa + 1}{2} \rfloor)$.
The density of the present graph $G_t$, $den_{G_t}(V(G_t)) \leq (\kappa + 1). den_{G'_t}(V(G_t)) + \frac{\kappa}{2}$.
\end{lemma}
\begin{IEEEproof}
%From Lemma~\cite{lemma: degree}, it follows that for the subset $S$, if each node had the maximum degree increase and all the added edges were part of the induced subgraph, then $E_{G_t}(S) = E_{G'_t}(S) +$
In the above statement, the l.h.s. represents the density of the present graph $G_t$ and the r.h.s. the density of the subgraph induced by the present graph $G_t$ on $G'_t$ (Note that since $G'_t$ suffers no deletions, $V(G_t) \subseteq V(G'_t)$). In Lemma~\ref{lemma: density-ub}, substituting $V(G_t)$ for $S$, we get, 
\[
%E_{G_t}(S) & \le & E_{G'_t}(S) +\sum_{i=1}^{|G_t|}\left(\frac{\kappa}{2} deg_{G'_t}(x_i) + \frac{\kappa}{2} \right)\\
den_{G_t}(V(G_t)) \leq  \frac{E_{G'_t}(V(G_t))|}{V(G_t)|}  +   \frac{\kappa.\sum_{i=1}^{|V(G_t)|} deg_{G'_t}(x_i)}{2|V(G_t)|} + \frac{\kappa}{2}\\
%den_{G_t}(V(G_t)) \leq  \frac{E_{G'_t}(V(G_t))|}{V(G_t)|}  +   \frac{\kappa.\sum_{i=1}^{|V(G_t)|} deg_{G'_t}(x_i), x_i \in V(G_t)}{2|V(G_t)|} + \frac{\kappa}{2}\\
 \]
 Since the number of edges,  $ |E_{G'_t}(V(G_t))|$ is $\frac{\sum_{i=1}^{|V(G_t)|} deg_{G'_t}(x_i)}{2}$, we get the lemma.
 % $den_{G_t}(V(G_t)) \leq \kappa. den_{G'_t}(V(G_t)) + \frac{\kappa}{2}$
%\end{eqnarray*}
\end{IEEEproof}

\subsection{Degree Analysis}

\begin{lemma}
\label{lemma: degree}
  For all $x \in G_t$, $degree_{G_t}(x) \le \kappa.degree_{G'_t}(x) + \kappa$,  for a fixed parameter $\kappa> 0$.
\end{lemma}

%\begin{lemma}
%\label{lemma: degreeApp}
%  For all $x \in G_t$, $degree_{G_t}(x) = O(\kappa.degree_{G'_t}(x))$,  for a fixed parameter $\kappa > 0$.
%\end{lemma}

\begin{proof}
 The proof is essentially the same as in \cite{PanduranganPODC11}. Recall that we call the original edges or the edges inserted by the adversary as black edges since they have the color black as part of their color set.  Intuitively, the same proof as that in xheal~\cite{PanduranganPODC11} holds since xheal does not depend on removing any black edges, and thus also covers the edge-preserving case. We give a brief proof counting the edges:
 
 We bound the increase in degree of any node $x$ that belongs to both $G_t$ and $G'_t$. 
 % Let the degree of $x$ in $G'_t$ be $d'(x) = 
The degree of $x$ in $G'_t$, $degree_{G'_t}(x)$, is simply the count of its black edges. There are three cases:
%This will be black-degree of $x$ (as $G'_t$ comprises solely of edges present
%in the original graph plus the inserted edges). There are three cases to consider and we bound the degree increase in each:

\begin{enumerate}
\item \emph{$x$ loses a black edge:} In xheal+, this can only happen when the adversary deletes a node (since xheal+ is edge preserving). Now, xheal+ may add upto $\kappa$  edges by making $x$ part of a primary or a secondary cloud. Thus, here,  $x$'s degree can increase by a factor of $\kappa$
%.
%in place of it, because
%a $\kappa$-regular expander is constructed which includes this node (this
%expander can be a primary or a secondary cloud). Thus $x$'s degree can increase
%by a factor of $\kappa$ at most because of deletion of black edges.

 \item \emph{$x$ loses a edge with a non-black color:}  This can happen due to adversarial node deletion or during reconstruction by the algorithm. This node deletion initiates a reconstruction of the $\kappa$-regular primary or secondary expander cloud (during which some non-black edges may be removed). At the end of this reconstruction, $x$ remains part of the $\kappa$ (or smaller)-regular degree cloud, and thus, does not increase its degree. Notice that only time the algorithm itself deletes non-black  edges is during reconstruction and does not add any more edges in lieu of these deleted edges.

%then the algorithm restructures the expander cloud by constructing a new $\kappa$-regular expander. Again, this is true if the reconstruction is done on a primary or a secondary cloud. In this case,
%the degree of $x$ does not change.  

\item \emph{$x$ becomes a bridge node:} This means that $x$ takes part in a secondary cloud for the first time, either for its own primary cloud or as a borrowed member by another cloud. Since the secondary cloud itself is a $\kappa$-regular expander, $x$ gains a degree of $\kappa$. However, $x$ never takes part in more than one secondary cloud.

% Finally, we consider the effect of non-free nodes. $x$'s degree can increase if it is chosen as a bridge (non-free) node to connect
%a primary cloud (with which it is associated) to a secondary cloud. In this case, its degree will increase by $\kappa$, since it will become part of the secondary cloud expander. There is one more possibility
%that can contribute to increase of $x$'s degree by $\kappa$ more. If $x$ is chosen to be shared as a free node, i.e., it gets associated as a free node with another primary cloud than it originally belongs to, then its degree increases by $\kappa$ more, since it becomes part of another $\kappa$-regular expander. The shared node
%becomes a bridge node, i.e., a non-free node in that time step. Hence it cannot
%be shared henceforth.
\end{enumerate}

Thus, from the above, we get that  $degree_{G_t}(x) \le \kappa.degree_{G'_t}(x) + \kappa$.
%$d(x) \leq \kappa  d'(x) + 2\kappa.$
%The lemma follows.
\end{proof}

\section{Other edge-preserving self-healing algorithms}
\label{sec: otheralgos}
 We look at two other recent self-healing algorithms, Forgiving Graph~\cite{HayesPODC09} and Forgiving Tree~\cite{HayesPODC08}. Our analysis shows that these algorithms are implicitly edge-preserving in the sense that they may not work properly if the algorithm ever deleted original or adversary added edges. At a high level, these algorithms have virtual nodes and edges, which are a counterpart of the clouds in xheal+, and real nodes and edges, which will be like black edges in xheal+.  The basic mechanism is to replace the deleted node by a 'Reconstruction Tree'  of virtual nodes and edges (i.e.  placeholder nodes simulated by the 'real' i.e. existing nodes in the network); this is shown in Figure~\ref{fig: RT} (from~\ref{HayesPODC09}. Without going into details, it is simple to use this mechanism of virtual structures instead of the marking scheme used in xheal+ to ensure edge-preservation if the underlying algorithm does not explicitly delete original edges.  Let us have a brief look at the measure of density (as an example of  an edge-monotonically non-decreasing property) for these algorithms.
 
 \begin{figure}[h!]
\centering
\includegraphics[scale=0.25]{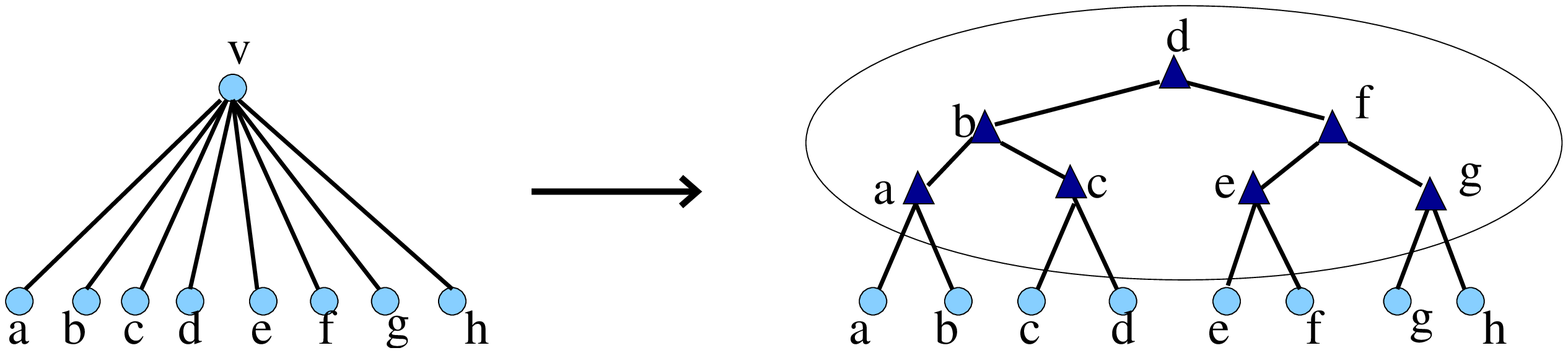}
\caption{Deleted node $v$ replaced by its Reconstruction Tree. The  triangle shaped nodes are 'virtual' helper nodes simulated by the 'real' nodes which are in the leaf layer.}
 \label{fig: RT}
\end{figure}

\subsection{Edge-preserving Forgiving Graph}
Forgiving Graph confirms to the general self-healing model(Figure~\ref{algo:model-2}) maintaining as it's success metrics connectivity, degree increase and network stretch. It has the following bound on degree increase (Theorem 1 from~\cite{HayesPODC09}):

\begin{lemma}
\label{lemma: fg-deg}
 For any node $v$ in $V(G_t)$, after any number of time steps, $t$,  $degree_ {G_t}(v) \le  3 degree_ {G'_t}(v)$.
\end{lemma}

The edge-preserving property and the above lemma yield (lemma and proof similar to Lemmas~\ref{lemma: density-lb}, \ref{lemma: density-ub}, and \ref{lemma: density-graph-ub}) :

\begin{lemma}
\label{lemma: fg-density}
%For all $S \subseteq V(G_t)$, $den_{G'_t}(S) \leq den_{G_t}(S) \leq max(\kappa.den_{G'_t}(S), \lfloor \frac{\kappa + 1}{2} \rfloor)$.
For all $S \subseteq V(G_t)$, $den_{G'_t}(S) \leq den_{G_t}(S) \leq den_{G'_t}(S) + \frac{3.\sum_{i=1}^{|S|} deg_{G'_t}(x_i)}{|S|} $.  For the graph density i.e. for $S = V(G_t)$, $den_{G'_t}(V(G_t)) \leq den_{G_t}(V(G_t)) \leq 3 den_{G'_t}(V(G_t)) $
\end{lemma}

\subsection{Edge-preserving Forgiving Tree}

Forgiving Tree is essentially a spanning tree maintenance algorithm and confirms to the general self-healing model(Figure~\ref{algo:model-2}) except that it does not handle node insertions. Thus, the comparison graph at any time $t$ $G'_t$ is the same as the initial graph $G_0$. Forgiving Tree maintains as its success metrics connectivity, degree increase and network stretch. It has the following bound on degree increase (Adapted from Theorem 1 of~\cite{HayesPODC09}):

\begin{lemma}
\label{lemma: fg-deg}
 For any node $v$ in $V(G_t)$, after any number of time steps, $t$,  $degree_ {G_t}(v) \le  degree_ {G'_t}(v) + 3$.
\end{lemma}

The edge-preserving property and the above lemma yield :

\begin{lemma}
\label{lemma: fg-density}
%For all $S \subseteq V(G_t)$, $den_{G'_t}(S) \leq den_{G_t}(S) \leq max(\kappa.den_{G'_t}(S), \lfloor \frac{\kappa + 1}{2} \rfloor)$.
%For all $S \subseteq V(G_t)$, $den_{G'_t}(S) \leq den_{G_t}(S) \leq \frac{\sum_{i=1}^{|S|} deg_{G'_t}(x_i)}{|S|}  + 3$.  For the graph density i.e. for $S = V(G_t)$, $den_{G'_t}(V(G_t)) \leq den_{G_t}(V(G_t)) \leq  den_{G'_t}(V(G_t)) + 3$
For all $S \subseteq V(G_t)$, $den_{G'_t}(S) \leq den_{G_t}(S) \leq den_{G'_t}(S)  + \frac{3}{2}$.  For the graph density i.e. for $S = V(G_t)$, $den_{G'_t}(V(G_t)) \leq den_{G_t}(V(G_t)) \leq  den_{G'_t}(V(G_t)) + \frac{3}{2}$.
\end{lemma}

\section{Conclusion}

% Conclusion and future work.

We have presented an efficient,  distributed algorithm  that withstands repeated adversarial node insertions and deletions by adding a small number of new edges after each deletion. It maintains key global invariants of the network while doing only localized changes and using only local information. Furthermore, it is edge-preserving, i.e. does not require any of the original edges to be deleted during any healing phase. This is a novel addition to all previous work and yields several desirable properties as a consequence including preserving subgraph densities. In addition, the algorithm maintains all previously studied global invariants. Firstly, assuming  the initial network was connected, the network stays connected. Secondly, the (edge) expansion of the network is at least as good as the expansion would have been without any adversarial  deletion, or is at least a constant. Thirdly,  the distance between any pair of nodes never increases by more than a $O(\log n)$ multiplicative factor than what the distance would be without the adversarial deletions.  Lastly, the above global invariants are achieved while not allowing the degree of any node to increase by more than a small multiplicative factor.

Our work opens a new line of work towards obtaining healing algorithms that respect some initial structure - beyond just certain global measures. We have shown that edge-preservation is possible and that this leads to several desirable properties. Can we go beyond this to maintain even more properties of the initial graph, such as the spectrum (with some slack)? What about preserving some notion of proximity sketches with nodes? This seems to open a new line of work. Further, the goal of maintaining edge-preserving was two fold: First to obtain structural guarantees on local and global properties, and second to minimize the cost of modifications (termination or initiation of new communication edges). For the first, can we reach a theoretical characterization of  what network properties are amenable to self-healing, especially, global properties which can be maintained by local changes? What about combinations of desired network invariants? For the latter, can the costs be modeled in  more robust and direct manner? Another interesting orthogonal question is whether there are deterministic algorithms that can yield the same bounds on all metrics as the current randomized healing algorithm. We can also extend the work to different models and domains. We can look at designing algorithms for less flexible networks such as sensor networks, explore healing with non-local edges or more complex notions of failure.


%\bibliography{selfheal} 
%\bibliographystyle{IEEEtran}
%\bibliographystyle{plain}

\end{document}